\documentclass[11pt,twoside]{article}
\pdfoutput=1

\usepackage{amsfonts}
\usepackage{amsmath}
\usepackage{amsthm}
\usepackage{amssymb}
\usepackage{amscd}
\usepackage{eucal}
\usepackage{bbm}
\usepackage{dsfont}

\usepackage{graphicx}

\usepackage{a4}
\usepackage{microtype}
\usepackage{times}
\usepackage{cite}

\input{macro.lib}

\renewcommand{\appendix}{%
   \renewcommand{\section}{
        \secdef\Appendix\sAppendix}%
   \setcounter{section}{0}%
   \renewcommand{\thesection}{\Alph{section}}%
   \renewcommand{\theequation}{\thesection.\arabic{equation}}%
}

\newcommand{\Appendix}[2][?]{%
     \refstepcounter{section}%
     \setcounter{equation}{0}%
     \addcontentsline{toc}{appendix}%
          {\protect\numberline{\appendixname~\thesection} #1}%
     \vspace{\baselineskip}%
     {\noindent\large\bfseries\appendixname\ \thesection: #2\par}%
     \sectionmark{#1}\vspace{\baselineskip}}

\newcommand{\sAppendix}[1]{%
     {\noindent\large\bfseries\appendixname\:: #1\par}%
     \sectionmark{#1}\vspace{\baselineskip}}

\allowdisplaybreaks

\begin{document}

\thispagestyle{empty}

\begin{center}

{\Large \bf
Late-time large-distance asymptotics of the transverse correlation
functions of the XX chain in the space-like regime}

\vspace{10mm}

{\large
Frank G\"{o}hmann,$^\dagger$
Karol K. Kozlowski$^\ast$ and Junji Suzuki$^\ddagger$}\\[3.5ex]
$^\dagger$Fakult\"at f\"ur Mathematik und Naturwissenschaften,\\
Bergische Universit\"at Wuppertal,
42097 Wuppertal, Germany\\[1.0ex]
$^\ast$Univ Lyon, ENS de Lyon, Univ Claude Bernard,\\ CNRS,
Laboratoire de Physique, F-69342 Lyon, France\\[1.0ex]
$^\ddagger$Department of Physics, Faculty of Science, Shizuoka University,\\
Ohya 836, Suruga, Shizuoka, Japan

\vspace{40mm}

{\large {\bf Abstract}}

\end{center}

\begin{list}{}{\addtolength{\rightmargin}{9mm}
               \addtolength{\topsep}{-5mm}}
\item
We derive an explicit expression for the leading term in the late-time,
large-distance asymptotic expansion of a transverse dynamical two-point
function of the XX chain in the spacelike regime. This expression is valid
for all non-zero finite temperatures and for all magnetic fields below
the saturation threshold. It is obtained here by means of a straightforward
term-by-term analysis of a thermal form factor series, derived in previous
work, and demonstrates the usefulness of the latter.
\end{list}

\clearpage

\section{Introduction}
The XX chain is a spin chain with Hamiltonian \cite{LSM61}
\begin{equation} \label{hxx}
     H_L = J \sum_{j = 1}^L \bigl( \s_{j-1}^x \s_j^x + \s_{j-1}^y \s_j^y \bigr)
           - \frac{h}{2} \sum_{j=1}^L \s_j^z \epc
\end{equation}
where the $\s_j^\a$, $\a = x, y, z$, are Pauli matrices acting on site
$j \in \{1, \dots, L\}$ of an $L$-site periodic lattice, $\s_0^\a = \s_L^\a$.
The parameters $J > 0$ and $h > 0$ denote the strengths of the spin-spin
interaction and of the applied magnetic field. We shall restrict the magnetic
field to values below the saturation threshold, $0 < h < 4J$.

In our recent work \cite{GKKKS17} we have derived a novel form factor series for
the transverse dynamical correlation function
\begin{equation} \label{defcorrmp}
     \bigl\< \s_1^- \s_{m+1}^+ (t) \bigr\>_T =
        \lim_{L \rightarrow + \infty}
	\frac{\tr \{ \re^{- H_L/T} \s_1^- \re^{\i H_L t} \s_{m+1}^+ \re^{- \i H_L t} \}}
	     {\tr \{ \re^{- H_L/T}\}}
\end{equation}
of the XX chain in equilibrium with a heat bath at temperature $T$.
It measures the space-time evolution of a local perturbation relating
two points at distance $m$ and temporal separation~$t$. Our series
originates from a form factor expansion related to the quantum
transfer matrix \cite{DGK13a}. It can be resummed into a `Fredholm
determinant representation' consisting of a prefactor times a Fredholm
determinant of an integrable integral operator \cite{GKS19app}. The
latter is different from the Fredholm determinant representation
derived by Colomo et al.\ in \cite{CIKT93}.

For Fredholm determinants and resolvent kernels of integrable integral
operators a general method \cite{DeZh93} is available that allows one
to analyse their asymptotic dependence on parameters. Starting with
the Fredholm determinant representation obtained in \cite{CIKT93} the authors
of \cite{IIKS93b} applied this `nonlinear steepest-descent method'
to the late-time, large-distance analysis of (\ref{defcorrmp}) at a
fixed ratio $\a = m/(4Jt)$. They found exponential decay of the form
\begin{equation} \label{symbasy}
     \bigl\< \s_1^- \s_{m+1}^+ (t) \bigr\>_T \sim C t^\nu \re^{- m/\x} \epc
\end{equation}
where $C$, $\nu$ and $\x$ depend on $T$, $h$ and $\a$. The
functional dependence differs according to whether $\a > 1$ or
$\a < 1$. The former regime, in which the spatial distance in units
of $4J$ is larger than the temporal separation, is called `spacelike',
while the latter is referred to as `the timelike regime'.

In \cite{IIKS93b} the authors considered magnetic fields below the
saturation threshold, $0 < h < 4J$. They obtained explicit expressions
for $\nu$ and $\x$ in both, space- and timelike regimes. Later the
`constant term' $C$ was obtained for $h > 4J$ in \cite{Jie98}. Although
the nonlinear steepest descent method would allow one to calculate
$C$ for $0 < h < 4J$ as well, it seems that nobody has ever attempted
to do so. This may be partially attributed to the cumbersome nature
of the required calculations. 

\enlargethispage{2ex}

In this work we reconsider the late-time, large-distance asymptotic
analysis of the two-point function $\<\s_1^- \s_{m+1}^+ (t) \bigr\>_T$
in the spacelike regime. It turns out that the novel thermal form
factor series derived in \cite{GKKKS17} allows us to obtain the
asymptotics, including the constant term $C$, by a rather elementary
term-by-term analysis of the series that avoids the use of any
Riemann-Hilbert problem.

On the other hand, our thermal form factor series can be resummed
into a Fredholm determinant representation as well. As we shall
see below this Fredholm determinant representation is rather
different from the one of Its et al.\ \cite{IIKS93b} in that the
term that provides the leading late-time, large-distance asymptotics
in the spacelike regime appears to be pulled out as a pre-factor.
Our finding strikingly resembles in structure the Borodin-Okounkov,
Geronimo-Case formula \cite{BoOk00,GeCa79,BaWi00} for a Toeplitz
determinant generated by a symbol satisfying the hypotheses of the
Szeg\"o theorem.

We should point out that the late-time, large-distance
asymptotics considered here do not commute with the low and
high-temperature asymptotics. At any finite temperature the
asymptotic decay of the transverse two-point functions is exponential
and given by~(\ref{symbasy}). If, however, the temperature is
send to zero first, the correlation functions will vary algebraically 
\cite{Kozlowski19pp}. We shall consider this limit for the more
general XXZ chain in subsequent work. If we send the temperature
to infinity first, then the behaviour of the correlation functions
in `time-direction' becomes Gaussian \cite{BrJa76,PeCa77}. We have
recently analysed the latter situation in full generality in
\cite{GKS19app}, which is one of two companion papers of this
work. In the other one \cite{GKSS19pp} we evaluate the two-point 
function numerically, for a wide range of temperature and
space-time separations, directly from the novel Fredholm determinant
representation.

\section{Thermal form factor series representation} \label{sec:ffseries}
The starting point of our analysis will be a thermal form factor series
for the transversal two-point function (\ref{defcorrmp}) derived in
\cite{GKKKS17}. The series is a series of multiple integrals which is
most compactly expressed in terms of certain functions characteristic
of the XX chain. These are in first place the momentum $p$ and the
energy $\eps$ of the single-particle excitations of the Hamiltonian
expressed in terms of the rapidity variable,
\begin{equation}
     p(\la) = - \i \ln \bigl( - \i \tgh(\la) \bigr) \epc \qd
     \eps (\la) = h + 2 J p' (\la) \epp
\end{equation}
Here we choose the principal branch of the logarithm in the definition
of the momentum function $p(\la)$, cutting the complex plane from
$- \i \p/2$ to zero modulo $\i \p$. Because of the $\p \i$-periodicity
of the momentum, shared by all other functions in our form factor series,
we may think of these functions as being defined on a cylinder of
circumference $\p$, which is equivalent to restricting their values to the
`fundamental strip'
\begin{equation}
     {\cal S} = \Bigl\{ \la \in {\mathbb C} \Big|
                        - \frac{\p}{4} \le \Im \la < \frac{3 \p}{4} \Bigr\} \epp
\end{equation}

It is easy to see that $\eps$ has precisely two roots
\begin{equation}
     \la_F^\pm = \frac{\i \p}{4} \pm z_F \epc \qd
           z_F = \2 \arch \biggl( \frac{4 J}{h} \biggr)
\end{equation}
in ${\cal S}$. These roots are called the Fermi rapidities.
The value
\begin{equation} \label{defpf}
     p_F = p(\la_F^-) = \arccos \biggl( \frac{h}{4J} \biggr)
\end{equation}
of the momentum function evaluated at the left Fermi rapidity is the 
Fermi momentum. Using the Fermi rapidities we can represent the
energy function as
\begin{equation} \label{epsfactor}
     \eps(\la) = - h \, p'(\la) \sh(\la - \la_F^-) \sh(\la - \la_F^+) \epp
\end{equation}
Energy and momentum functions $\eps$ and $p$ are real on the lines
$x \pm \i \p/4$, $x \in {\mathbb R}$, where they take the values
\begin{subequations}
\label{pepsreal}
\begin{align}
     & \eps(x \pm \i \p/4) = h \mp \frac{4J}{\ch (2 x)} \epc \\[1ex]
     & p(x + \i \p/4) = - \frac{\p}{2} + 2 \arctg \bigl( \re^{- 2 x} \bigr) \epc \\
     & p(x - \i \p/4) = - \p \sign(x) + \frac{\p}{2} - 2 \arctg \bigl( \re^{-2 x} \bigr) \epp
\end{align}
\end{subequations}

The one-particle energy determines the function
\begin{equation} \label{funz}
     z(\la) = \frac{1}{2 \p \i} \ln \biggl[ \cth \biggl(\frac{\eps(\la)}{2 T} \biggr) \biggr]
	      \epp
\end{equation}
Most of the functions occurring in the form-factor series below are
defined as integrals over two simple closed contours ${\cal C}_h$
and ${\cal C}_p$, involving $p$, $\eps$, $z$ and some hyperbolic
functions.
\begin{figure}
\begin{center}
\includegraphics[width=.93\textwidth]{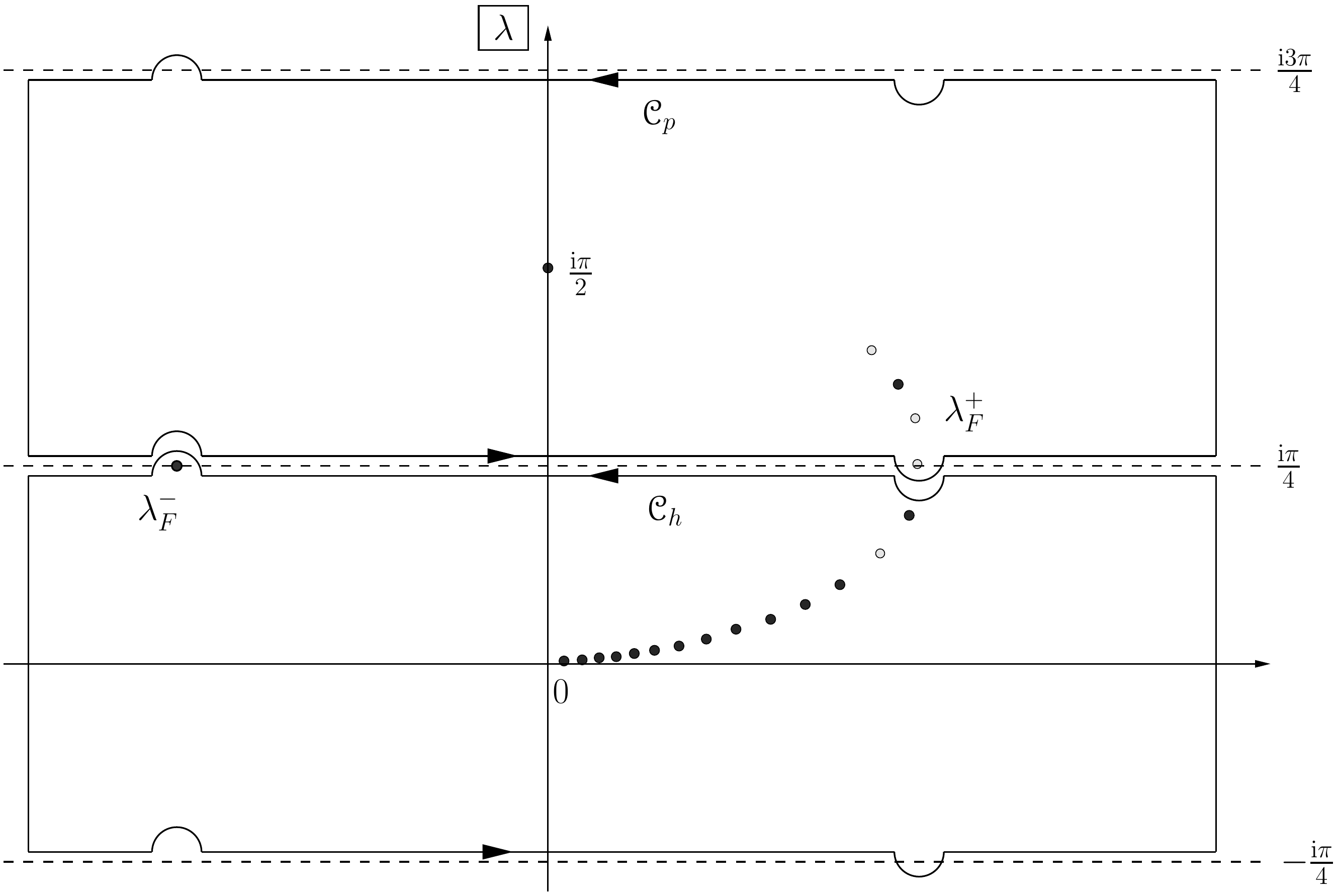}
\caption{\label{fig:xx_ch_and_cp} Sketch of the hole and
particle contours ${\cal C}_h$ and ${\cal C}_p$. The Fermi rapidity $\la_F^-$
is located inside ${\cal C}_h$, while $\la_F^+$ lies inside ${\cal C}_p$.}
\end{center}
\end{figure}
The `hole contour' ${\cal C}_h$ and the `particle contour' ${\cal C}_p$
are sketched in Fig.~\ref{fig:xx_ch_and_cp}. They are defined in such a way
that ${\cal C}_h$ encloses all roots of $\re^{- \eps(x)/T} - 1$ located
inside the strip $- \p/4 < \Im x < \p/4$ (`the holes') as well as the
left Fermi rapidity $\la_F^-$, whereas ${\cal C}_p$ encloses the roots of
$\re^{- \eps(x)/T} - 1$ inside the strip $\p/4 < \Im x < 3 \p/4$ (`the
particles') as well as the right Fermi rapidity $\la_F^+$.

Given these contours we define the Cauchy transforms
\begin{equation} \label{phh}
     \PH_h (x) = \frac{\i p'(x)}{2}
                 \exp \biggl\{ \i \int_{{\cal C}_h} \rd \la \: p'(\la) z(\la)
	                       \frac{\sh(x + \la)}{\sh(x - \la)} \biggr\}
\end{equation}
for all $x \in {\cal S} \setminus {\cal C}_h$, and
\begin{equation} \label{php}
     \PH_p (x) = \frac{\i p'(x)}{2}
                 \exp \biggl\{ - \i \int_{{\cal C}_p} \rd \la \: p'(\la) z(\la)
	                       \frac{\sh(x + \la)}{\sh(x - \la)} \biggr\}
\end{equation}
for all $x \in {\cal S} \setminus {\cal C}_p$.
For fixed $x \in \Int ({\cal C}_h) \cup \Int ({\cal C}_p)$ the function
$\sh(x + \la)/\sh(x - \la)$ is holomorphic in $\la$ for all $\la \in
{\cal S} \setminus \bigl(\Int ({\cal C}_h) \cup \Int ({\cal C}_p)\bigr)$.
Since the integrands in (\ref{phh}), (\ref{php}) are rapidly decaying
for $\la \rightarrow \pm \infty$, we may deform the contours and conclude
that
\begin{equation}
     \PH_h (x) = \PH_p (x) \qd
        \text{for all $x \in \Int ({\cal C}_h) \cup \Int ({\cal C}_p$)} \epp
\end{equation}
Another function needed below is the square of a generalized Cauchy
determinant,
\begin{equation}
     {\cal D} \bigl(\{x_j\}_{j=1}^m, \{y_k\}_{k=1}^n\bigr) =
        \frac{\bigl[ \prod_{1 \le j < k \le m} \sh^2 (x_j - x_k) \bigr]
              \bigl[ \prod_{1 \le j < k \le n} \sh^2 (y_j - y_k) \bigr]}
             {\prod_{j=1}^m \prod_{k=1}^n \sh^2 (x_j - y_k)} \epp
\end{equation}

After these preparations we can now recall the form factor series derived
in \cite{GKKKS17}. Using the above notation and performing several more or less 
obvious simplifications it can be written as
\begin{multline} \label{ffseriesxxtrans}
     \bigl\<\s_1^- \s_{m+1}^+ (t)\bigr\>_T = (-1)^m {\cal F} (m)
        \sum_{n=1}^\infty \frac{(-1)^n}{n! (n-1)!}
	   \prod_{j=1}^n \int_{{\cal C}_h} \frac{\rd x_j}{\p \i}
	      \frac{\PH_p (x_j) \re^{\i (m p(x_j) - t \eps(x_j))}}{1 - \re^{\eps (x_j)/T}}
	      \\ \times
	   \prod_{k=1}^{n-1} \int_{{\cal C}_p} \frac{\rd y_k}{\p \i}
	      \frac{\re^{- \i (m p(y_k) - t \eps(y_k))}}
	           {\PH_h (y_k) \bigl(1 - \re^{- \eps(y_k)/T}\bigr)} \:
		    {\cal D} \bigl( \{x_j\}_{j=1}^n, \{y_k\}_{k=1}^{n-1} \bigr) \epc
\end{multline}
where
\begin{multline} \label{deff}
     {\cal F} (m) = \re^{- \i m p_F}
        \exp \biggl\{- \int_{{\cal C}_h' \subset {\cal C}_h} \rd \la \: z(\la)
	\int_{{\cal C}_h} \rd \m \: \cth' (\la - \m) z(\m) \biggr\}\\
	\times \exp \biggl\{- m \int_{{\cal C}_h} \frac{\rd \la}{2 \p} p' (\la)
	\ln \biggl| \cth \biggl(\frac{\eps(\la)}{2 T} \biggr) \biggr| \biggr\} \epp
\end{multline}
The contour ${\cal C}_h'$ is tightly enclosed by ${\cal C}_h$.

\section{Asymptotics in the spacelike regime}
\begin{theorem*}
In the spacelike regime $m > 4Jt$ the form factor series
(\ref{ffseriesxxtrans}) is absolutely convergent and determines
the late-time, large-distance asymptotics of the transverse dynamical
correlation function of the XX chain as
\begin{multline} \label{xxtransasy}
     \bigl\<\s_1^- \s_{m+1}^+ (t)\bigr\>_T = C(T,h) (-1)^m
        \exp \biggl\{
	   - m \int_{{\cal C}_h} \frac{\rd \la}{2 \p} p' (\la)
	   \ln \biggl| \cth \biggl(\frac{\eps(\la)}{2 T} \biggr) \biggr| \biggr\} \\[1ex]
	   \times \bigl(1 + \CO(t^{-\infty})\bigr) \epc
\end{multline}
where
\begin{equation} \label{asyconstant}
     C(T,h) = \frac{2 T \PH_p (\la_F^-)}{\eps' (\la_F^-)}
              \exp \biggl\{- \int_{{\cal C}_h' \subset {\cal C}_h} \rd \la \: z(\la)
	                     \int_{{\cal C}_h} \rd \m \: \cth' (\la - \m) z(\m) \biggr\} \epp
\end{equation}
\end{theorem*}

In preparation of the proof we introduce the short-hand notations
\begin{equation}
     \tau = 4Jt \epc \qd \a = \frac{m}{\tau}
\end{equation}
and the function
\begin{equation}
     g(\la) = \i \bigl(\a p (\la) + \cos (p(\la))\bigr)
\end{equation}
with real and imaginary parts $u(\la) = \Re g(\la)$ and
$v(\la) = \Im g(\la)$. Then the `wave factors' in
(\ref{ffseriesxxtrans}) take the form
\begin{equation}
     \re^{\pm \i (m p(\la) - t \eps (\la))} =
        \re^{\mp \i h t \pm \tau g(\la)} \epp
\end{equation}

We will be interested in the asymptotic behaviour of
(\ref{ffseriesxxtrans}) for large positive $\tau$ and fixed
$\a > 1$. As we shall see below, it is determined by the
poles of the integrands at $\la_F^\pm$. The saddle points
contribute only to the subleading corrections. This becomes
clear when we consider the function $g$ close to the lines
${\mathbb R} \pm \i \p/4$ and on the lines $\Re \la = \pm R$
for $R > 0$ large enough.

\begin{lemma*}
Fix $\a > 1$.
\begin{enumerate}
\item
Then $g' (\la) \ne 0$ for all $\la \in {\mathbb R} \pm \i \p/4 \mod \i \p$, i.e.\
there are no saddle points on these lines.
\item
Define the oriented contours
\begin{align} \label{defseccont}
     & \chsd = \Bigl[- R + \frac{\p \i}4 - \i \de, - R - \frac{\p \i}4 + \i \de \Bigr] \cup
               \Bigl[- R - \frac{\p \i}4 + \i \de, R - \frac{\p \i}4 + \i \de \Bigr]
	       \notag \\ & \mspace{69.mu} \cup
               \Bigl[R - \frac{\p \i}4 + \i \de, R + \frac{\p \i}4 - \i \de \Bigr] \cup
               \Bigl[R + \frac{\p \i}4 - \i \de, - R + \frac{\p \i}4 - \i \de \Bigr] \epc
	       \notag \\[1ex]
     & \cpsd = \chsd + \frac{\p \i}{2} \epc
\end{align}
where $R, \de > 0$. Then $R$ and $\de$ can be chosen in such a way that
$u(\la) < 0$ for all $\la \in \chsd$, $u(\la) > 0$ for all $\la \in \cpsd$
and all hole roots are inside $\chsd$, while all particle roots are inside
$\cpsd$.
\end{enumerate}
\end{lemma*}
\begin{proof}
(i) For all $\a > 1$ and $\la \in {\mathbb R} \pm \i \p/4$ we have
\begin{equation} \label{derg}
     g' (\la) = \i p'(\la) \bigl(\a - \sin(p(\la))\bigr) \ne 0 \epc
\end{equation}
since
\begin{equation} \label{derip}
     \i p'(\la) = \frac{2}{\sh(2 \la)} \ne 0
\end{equation}
for all $\la \in {\cal S}$, and $p (\la) \in {\mathbb R}$
for all $\la \in {\mathbb R} \pm \i \p/4$ (see (\ref{pepsreal})).

(ii) Let $x = \Re \la$, $y = \Im \la$. Due to (\ref{derg}), (\ref{derip})
and the Cauchy-Riemann equations
\begin{equation} \label{derrealpartg}
     \6_y u (\la) = - \Im g'(\la) =
        \pm \frac{2}{\ch(2x)} \bigl(\a - \sin(p(x \pm \i\p/4))\bigr)
\end{equation}
for $\la = x \pm \i\p/4$. Now $\a > 1$ by assumption. Thus,
(\ref{derrealpartg}) implies that
\begin{equation} \label{signdu}
     \6_y u (\la)
        \begin{cases}
	   > 0 & \text{for $\la \in {\mathbb R} + \i\p/4$} \\
	   < 0 & \text{for $\la \in {\mathbb R} - \i\p/4$.}
        \end{cases}
\end{equation}
Since $u = 0$ for $\la \in {\mathbb R} \pm \i\p/4$ it follows that
$u (\la) < 0$ on the lines $\la \in {\mathbb R} \pm \i\p/4 \mp \i \de$
for small enough positive $\de$. Similarly, $u(\la) > 0$ on the lines
${\mathbb R} + \i\p/4 + \i \de$ and ${\mathbb R} + 3\i\p/4 - \i \de$.

Since $\a > 1$, there is a unique $\ph > 0$ such that $\a = \cth (2\ph)$.
Using this parameterization we find for any $\la = x + \i y \in {\cal S}$
that
\begin{multline}
     \6_y u(\la) = \frac{4}{\sh(2 \ph) |\ch(4\la) - 1|^2}
                   \bigl[\sh(4x) \sh(2(x + \ph)) \cos(2y) \sin(4y) \\
		         - (\ch(4x) \cos(4y) - 1) \ch(2(x + \ph)) \sin(2y)\bigr] \epp
\end{multline}
Thus, $\6_y u (\la) = 0$ if and only if
\begin{equation}
     \sin(2y) \ch(2(x - \ph))
        \biggl[\frac{\sh(4x) \sh(2(x + \ph)) + \ch(2(x + \ph))}{\ch(2(x - \ph))}
               - \cos(4y)\biggr] = 0 \epp
\end{equation}
Here the first term in the square bracket is unbounded from above
for $x \rightarrow \pm \infty$, implying that the only roots of
$\6_y u(\la)$ in $\cal S$ are at $y = 0, \p/2$ if $|x|$ is large enough.
Taking into account (\ref{signdu}) we see that, if the latter is the
case, then
\begin{equation} \label{signdu2}
     \6_y u (\la)
        \begin{cases}
	   > 0 & \text{for $y \in (0,\p/2)$} \\
	   < 0 & \text{for $y \in (- \p/4,0) \cup (\p/2,3\p/4)$.}
        \end{cases}
\end{equation}
It follows that $u(\la) < 0$ of the lines $\pm R + \i (-\p/4, \p/4)$,
while $u(\la) > 0$ on $\pm R + \i (\p/4, 3 \p/4)$, if $R > 0$ large
enough. The statement about the location of the particle and hole roots
follows by straightforward inspection of the integrands in
(\ref{ffseriesxxtrans}).
\end{proof}

\begin{proof}[Proof of the Theorem]
The function ${\cal D} \bigl(\{x_j\}_{j=1}^n, \{y_k\}_{k=1}^{n-1}\bigr)$
is symmetric separately in all $x_j$ and $y_k$. It satisfies
\begin{equation}
     {\cal D} \bigl(\{x_j\}_{j=1}^n, \{y_k\}_{k=1}^{n-1}\bigr) = 0
\end{equation}
if $x_j = x_k$ or $y_j = y_k$ for all $j \ne k$. Setting
\begin{equation}
     V_h (x) = \frac{\PH_p (x) \re^{\i (m p(x) - t \eps(x))}}
	            {\p \i \bigl(\re^{\eps (x)/T} - 1\bigr)} \epc \qd
     V_p (x) = \frac{\re^{- \i (m p(y) - t \eps(y))}}
	            {\p \i \PH_h (y) \bigl(1 - \re^{- \eps (y)/T}\bigr)}
\end{equation}
and using the above lemma we therefore obtain
\begin{align} \label{ffseriessplit}
     & \bigl\<\s_1^- \s_{m+1}^+ (t)\bigr\>_T = (-1)^m {\cal F} (m)
        \sum_{n=1}^\infty \frac{1}{n! (n-1)!}
	   \int_{{\cal C}_h^n} \rd^n x \: \biggl[ \prod_{j=1}^n V_h (x_j) \biggr]
	   \notag \\[-.5ex] & \mspace{180.mu} \times
	   \int_{{\cal C}_p^{n-1}} \rd^{n-1} y \: \biggl[ \prod_{k=1}^{n-1} V_p (y_k) \biggr]
		    {\cal D} \bigl( \{x_j\}_{j=1}^n, \{y_k\}_{k=1}^{n-1} \bigr)
        \notag \\[.5ex] & \mspace{18.mu}
        = (-1)^m {\cal F} (m) \sum_{n=1}^\infty \frac{1}{n! (n-1)!} \biggl(
	   \int_{\chsd^n} \rd^n x \: \prod_{j=1}^n V_h (x_j)
	   \notag \\[-.5ex] & \mspace{162.mu}
	   + n \int_{\chsd^{n-1}} \rd^{n-1} x \: \biggl[ \prod_{j=1}^{n-1} V_h (x_j) \biggr]
	   2 \p \i \res \bigl\{ \rd x_n \: V_h (x_n), x_n = \la_F^- \bigr\} \biggr)
	   \notag \\[.5ex] & \mspace{45.mu} \times \biggl(
	   \int_{\cpsd^{n-1}} \rd^{n-1} y \: \prod_{k=1}^{n-1} V_p (y_k)
	   \notag \\[-.5ex] & \mspace{72.mu} + (n - 1)
	   \int_{\cpsd^{n-2}} \rd^{n-2} y \: \biggl[ \prod_{k=1}^{n-2} V_p (y_k) \biggr]
	   2 \p \i \res \bigl\{ \rd y_{n-1} \: V_p (y_{n-1}), y_{n-1} = \la_F^+ \bigr\} \biggr)
	   \notag \\ & \mspace{360.mu}
	   \times {\cal D} \bigl(\{x_j\}_{j=1}^n, \{y_k\}_{k=1}^{n-1}\bigr) \epp
\end{align}
Here $\chsd$ and $\cpsd$ are the contours introduced in (\ref{defseccont}).
Notice that we consider $\res \bigl\{ \rd x \: V_h (x), x = \la_F^- \bigr\}$
as a functional acting on functions $f$ holomorphic in a disc $D_\e (\la_F^-)$
of sufficiently small radius $\e$ centered about $\la_F^-$ as
\begin{equation}
     \res \bigl\{ \rd x \: V_h (x), x = \la_F^- \bigr\} f
        = \int_{D_\e (\la_F^-)} \frac{\rd x}{2 \p \i} \: V_h (x) f(x) \epc
\end{equation}
and similarly for $\res \bigl\{ \rd y \: V_p (y), y = \la_F^+ \bigr\}$.
In particular,
\begin{subequations}
\begin{align}
     & 2 \p \i \res \bigl\{ \rd x \: V_h (x), x = \la_F^- \bigr\} 1
        = \frac{2 T \re^{\i m p_F} \PH_p (\la_F^-)}{\eps' (\la_F^-)} \epc \\
     & 2 \p \i \res \bigl\{ \rd y \: V_p (y), y = \la_F^+ \bigr\} 1
        = \frac{2 T \re^{\i m p_F}}{\PH_h (\la_F^+) \eps' (\la_F^+)} \epp
\end{align}
\end{subequations}

Equation (\ref{ffseriessplit}) implies that
\begin{equation} \label{sums}
     \bigl\<\s_1^- \s_{m+1}^+ (t)\bigr\>_T
        = (-1)^m {\cal F} (m) \sum_{\ell = 1}^4 S_\ell (m,t) \epc
\end{equation}
where the four series $S_\ell (m,t)$ can be written as follows.
\begin{subequations}
\begin{align}
     & S_1 (m,t) = \frac{2 T \re^{\i m p_F} \PH_p (\la_F^-)}{\eps' (\la_F^-)}
                   \sum_{n=0}^\infty S_1^{(n)} (m,t) \epc \\[1ex]
     & S_2 (m,t) = - \frac{1}{\sh^2 (\la_F^+ - \la_F^-)}
                   \biggl(\frac{2 T \re^{\i m p_F}}{\eps' (\la_F^-)}\biggr)^2
                   \frac{\PH_p (\la_F^-)}{\PH_h (\la_F^+)}
                   \sum_{n=0}^\infty S_2^{(n)} (m,t) \epc \\[1ex]
     & S_3 (m,t) = \frac{2 T \re^{\i m p_F}}{\PH_h (\la_F^+) \eps' (\la_F^-)}
                   \sum_{n=0}^\infty S_3^{(n)} (m,t) \epc \\[1ex]
     & S_4 (m,t) = \sum_{n=0}^\infty S_4^{(n)} (m,t)
\end{align}
\end{subequations}
with
\enlargethispage{-3ex}
\begin{subequations}
\begin{align}
     & S_1^{(n)} (m,t) = \frac{1}{(n!)^2}
	  \int_{\chsd^n} \mspace{-18.mu} \rd^n x \: \biggl[\prod_{j=1}^n V_h (x_j)\biggr]
	  \int_{\cpsd^n} \mspace{-18.mu} \rd^n y \: \biggl[\prod_{k=1}^n V_p (y_k)\biggr]
	  \notag \\ & \mspace{180.mu} \times
	  \biggl[\prod_{j=1}^n \frac{\sh^2(x_j - \la_F^-)}{\sh^2 (y_j - \la_F^-)}\biggr]
	  {\cal D} \bigl(\{x_j\}_{j=1}^n, \{y_k\}_{k=1}^n\bigr) \epc \\[1ex]
     & S_2^{(n)} (m,t) = \frac{1}{(n+1)!n!}
	  \int_{\chsd^{n+1}} \mspace{-18.mu} \rd^{n+1} x \:
	  \biggl[\prod_{j=1}^{n+1} V_h (x_j)
	         \frac{\sh^2(x_j - \la_F^-)}{\sh^2 (x_j - \la_F^+)}\biggr]
	  \notag \\ & \mspace{18.mu} \times
	  \int_{\cpsd^n} \mspace{-18.mu} \rd^n y \:
	  \biggl[\prod_{k=1}^n V_p (y_k)
	         \frac{\sh^2(y_k - \la_F^+)}{\sh^2 (y_k - \la_F^-)}\biggr]
	  \biggl[\prod_{j=1}^n
	         \frac{\sh^2(x_j - x_{n+1})}{\sh^2 (y_j - x_{n+1})}\biggr]
	  {\cal D} \bigl(\{x_j\}_{j=1}^n, \{y_k\}_{k=1}^n\bigr) \epc \\[1ex]
     & S_3^{(n)} (m,t) = - \frac{1}{(n+2)!n!}
	  \int_{\chsd^{n+2}} \mspace{-18.mu} \rd^{n+2} x \:
	  \biggl[\prod_{j=1}^{n+2} \frac{V_h (x_j)}{\sh^2(x_j - \la_F^+)} \biggr]
	  \sh^2 (x_{n+1} - x_{n+2})
	  \notag \\ & \times
	  \int_{\cpsd^n} \mspace{-20.mu} \rd^n y \:
	  \biggl[\prod_{k=1}^n V_p (y_k) \sh^2(y_k - \la_F^+)\biggr]
	  \biggl[\prod_{j=1}^n \prod_{k=1}^2
	         \frac{\sh^2(x_j - x_{n+k})}{\sh^2 (y_j - x_{n+k})}\biggr]
	  {\cal D} \bigl(\{x_j\}_{j=1}^n, \{y_k\}_{k=1}^n\bigr), \\[1ex]
     & S_4^{(n)} (m,t) = \frac{1}{(n+1)!n!}
	  \int_{\chsd^{n+1}} \mspace{-18.mu} \rd^{n+1} x \:
	  \biggl[\prod_{j=1}^{n+1} V_h (x_j)\biggr]
	  \int_{\cpsd^n} \mspace{-18.mu} \rd^n y \: \biggl[\prod_{k=1}^n V_p (y_k)\biggr]
	  \notag \\ & \mspace{180.mu} \times
	  \biggl[\prod_{j=1}^n \frac{\sh^2(x_j - x_{n+1})}{\sh^2 (y_j - x_{n+1})}\biggr]
	  {\cal D} \bigl(\{x_j\}_{j=1}^n, \{y_k\}_{k=1}^n\bigr) \epp
\end{align}
\end{subequations}

In order to show the convergence of the series and to estimate their
asymptotic behaviour, we have to establish bounds on the individual
terms. We start with the functions ${\cal D} \bigl(\{x_j\}_{j=1}^n,
\{y_k\}_{k=1}^n\bigr)$ and recall the Hadamard bound for the determinant
of an $n \times n$ matrix
\begin{equation} \label{hadamard}
     \Bigl|\det_{j,k = 1, \dots, n} (M_{jk}) \Bigr| \le
        \Bigl( \max_{j,k = 1, \dots, n} |M_{jk}| \Bigr)^n \cdot n^\frac{n}{2} \epp
\end{equation}
Since the contours $\chsd$ and $\cpsd$ are finite and disjoint, we can
use (\ref{hadamard}) to estimate
\begin{equation}
     \Bigl|{\cal D} \bigl(\{x_j\}_{j=1}^n, \{y_k\}_{k=1}^n\bigr)\Bigr| =
        \biggl| \det_{j,k = 1, \dots, n} \biggl(\frac{1}{\sh(x_j - y_k)}\biggr)\biggr|^2
	\le B^{2n} n^n \epc
\end{equation}
where
\begin{equation}
     B = \sup_{\substack{x \in \chsd \\ y \in \cpsd}}
         \biggl| \frac{1}{\sh(x - y)} \biggr| \epp
\end{equation}
Likewise we set
\begin{equation}
     C = \sup_{\substack{x \in \chsd \\ y \in \cpsd}}
         \biggl|\frac{\sh(x - \la_F^-)}{\sh(y - \la_F^-)}\biggr| \epp
\end{equation}
As follows from the above lemma, there exist $\k, c > 0$ such that
\begin{equation}
     \sup_{\substack{x \in \chsd \\ y \in \cpsd}}
        \bigl( \max \bigl\{ |V_h (x)|, |V_p (y)| \bigr\} \bigr) = \k \re^{- \tau c} \epp
\end{equation}
With this we obtain a bound on every individual term in the series $S_1$,
\begin{equation}
     |S_1^{(n)} (m,t)| \le \frac{1}{(n!)^2} \bigl(|\chsd||\cpsd|\bigr)^n C^{2n} \k^{2n}
                           \re^{- 2n \tau c} B^{2n} n^n
                       \le \frac{1}{n!} C_1^n  \re^{- 2n \tau c}
\end{equation}
for some constant $C_1 > 0$. This implies absolute convergence of the series $S_1$
and shows that, asymptotically for large $\tau$, the series behaves like
\begin{equation} \label{asys1}
     S_1 (m,t) = \frac{2 T \re^{\i m p_F} \PH_p (\la_F^-)}{\eps' (\la_F^-)}
                 \bigl(1 + {\cal O}(\re^{- 2 \tau c})\bigr) \epp
\end{equation}

In a similar way one obtains
\begin{subequations}
\begin{align}
     & |S_2^{(n)} (m,t)| \le \frac{1}{n!} C_2^n \re^{- (2n+1) \tau c} \epc \\[.5ex]
     & |S_3^{(n)} (m,t)| \le \frac{1}{n!} C_3^n \re^{- (2n+2) \tau c} \epc \\[.5ex]
     & |S_4^{(n)} (m,t)| \le \frac{1}{n!} C_4^n \re^{- (2n+1) \tau c} \epc
\end{align}
\end{subequations}
for constants $C_j > 0$, $j = 2, 3, 4$. It follows that the series $S_j$,
$j = 2, 3, 4$, converge absolutely and behave asymptotically as
\begin{subequations}
\label{asys234}
\begin{align}
     & S_2 (m,t) = {\cal O} (\re^{- \tau c}) \epc \\[.5ex]
     & S_3 (m,t) = {\cal O} (\re^{- 2 \tau c}) \epc \\[.5ex]
     & S_4 (m,t) = {\cal O} (\re^{- \tau c}) \epp
\end{align}
\end{subequations}
Inserting (\ref{asys1}), (\ref{asys234}) into (\ref{sums}) and recalling the
explicit form (\ref{deff}) of ${\cal F} (m)$ we have arrived at the statement
of the theorem.
\end{proof}
The theorem fixes the constant term of the asymptotics in 
the spacelike regime that remained undetermined in \cite{IIKS93b}.
Note that the function $\eps' (\la_F^-)$ can be easily calculated
explicitly,
\begin{equation}
     \eps' (\la_F^-) = - 2 h \sqrt{1 - \Bigl(\frac{h}{4J}\Bigr)^2} \epp
\end{equation}
For the other factors composing the constant $C(T,h)$ we did
not find any further simplification so far.

\section{Discussion}
For the interpretation of our result we would like to recall 
a Fredholm determinant representation of the transversal two-point
function (\ref{defcorrmp}) that was obtained in \cite{GKS19app},
where it was used for the asymptotic analysis of the correlation
function in the high-temperature limit. Referring to \cite{GKS19app}
we define the functions
\begin{equation}
     \ph(x,y) = \frac{\re^{y - x}}{\sh(y - x)}
\end{equation}
and
\begin{subequations}
\begin{align}
     & \Om = \int_{{\cal C}_h} \rd x \: V_h (x) \epc \\
     & E_h (x) = \int_{{\cal C}_h} \rd y \: V_h (y) \ph(y,x) \epc \\
     & V(x,y) = \int_{{\cal C}_h} \rd z \: V_h (z) \ph(z,x) \ph(z,y) \epp
\end{align}
\end{subequations}
Using these functions we define two integral operators $\widehat{V}$
and $\widehat{P}$ acting on functions on the contour ${\cal C}_p$,
\begin{subequations}
\begin{align}
     \widehat V f (x) & = \int_{{\cal C}_p} \rd y \: V_p (y) V(x,y) f(y) \epc \\
     \widehat P f (x) & = \frac{E_h(x)}{\Om} \int_{{\cal C}_p} \rd y \:
                             V_p (y) E_h (y) f(y) \epp
\end{align}
\end{subequations}
Then (cf.\ \cite{GKS19app}) the transversal correlation functions of
the XX chain admit the Fredholm determinant representation
\begin{equation} \label{freddetus}
     \bigl\<\s_1^- \s_{m+1}^+ (t)\bigr\>_T =
        (-1)^m {\cal F} (m) \Om (m,t)
	\det_{{\cal C}_p} \bigl(\id + \widehat V - \widehat P\bigr) \epp
\end{equation}

Comparing with the asymptotic behaviour of the correlation
function in the spacelike regime $m > 4Jt$ we see that
\begin{equation}
     \det_{{\cal C}_p} \bigl(\id + \widehat V - \widehat P\bigr)
        \sim 1 + {\cal O} (t^{-\infty}) \epc
\end{equation}
meaning that the Fredholm determinant collects the higher-order
corrections to the main asymptotics. This is the analogy with
the Borodin-Okounkov-Geronimo-Case formula \cite{BoOk00,GeCa79}
mentioned in the introduction.

On the level of the Fredholm determinant representation it is
easiest to compare our result with that of Its, Izergin, Korepin
and Slavnov \cite{IIKS93b}. For this purpose we rewrite their
integral operators acting on functions on the the unit circle as
integral operators acting on functions on ${\cal C} = [- \infty - \i \p/4,
+ \infty - \i \p/4] \cup [+ \infty + \i \p/4, - \infty + \i \p/4]$.
This is achieved by employing the map $z \mapsto \re^{\i p(\la)}$
to the Fredholm determinant representation in \cite{IIKS93b}.
Then
\begin{equation} \label{freddetiiks}
     \bigl\<\s_1^- \s_{m+1}^+ (t)\bigr\>_T =
        (-1)^m \bigl[ \det_{{\cal C}} \bigl(\id + \widehat W + \widehat Q\bigr)
	              - \det_{{\cal C}} \bigl(\id + \widehat W\bigr) \bigr] \epc
\end{equation}
where $\widehat W$ is an integrable operator with kernel
\begin{subequations}
\begin{align}
     W(\la, \m) & = \frac{\ch(\la) H(\la) - \ch(\m) H(\m)}{\sh(\la - \m)}
                    \frac{\re^{\i (m p(\m) - t \eps(\m))}}{\p (1 + \re^{\eps(\m)/T})} \epc \\
     H(\la) & = {\rm v.p.} \int_{\cal C} \frac{\rd \m}{\p}
                              \frac{\re^{- \i(m p(\m) - t \eps(\m))}}{\ch(\m) \sh(\m - \la)}
\end{align}
\end{subequations}
and $\widehat Q$ is a one-dimensional projector acting as
\begin{equation}
     \widehat Q f(\la) = \frac{1}{\ch(\la)} \int_{\cal C} \frac{\rd \m}{2\p \i}
                        \frac{\re^{\i(m p(\m) - t \eps(\m))}
			f(\m)}{\sh (\m) (1 + \re^{\eps(\m)/T})} \epp
\end{equation}
Comparing (\ref{freddetus}) and (\ref{freddetiiks}) we see that
in (\ref{freddetiiks}) the late-time, large-distance asymptotics
is entirely inside the Fredholm determinants and therefore harder
to analyse.

The fact that the late-time, large-distance asymptotic behaviour of
the transverse dynamical correlation functions of the XX chain,
including the constant term, can be obtained directly from the
series representation (\ref{ffseriesxxtrans}) raises a number of
interesting questions.
\begin{enumerate}
\item
Is a similar analysis possible for the XXZ quantum spin chain?
Unlike the XX chain treated in this work no Fredholm determinant
representation for its two-point function is expected to exist,
but a thermal form-factor series similar to (\ref{ffseriesxxtrans})
is still available \cite{GKKKS17}. As the structure of the
saddle-point equations is very similar, there seems to be a good
chance that the answer will turn out to be positive.
\item
What can be done in the timelike regime? Here all terms in the
series (\ref{ffseriesxxtrans}) contribute to the late-time,
large-distance asymptotics. A further resummation would be necessary.
Can we devise a method to find such a generalization?
\end{enumerate}

We would like to close with two remarks. First, in our recent work
\cite{GKSS19pp} we have compared the asymptotic formula of our
theorem with a numerical evaluation based on the Fredholm determinant
representation (\ref{freddetus}). As should be clear from the fact
that the corrections are exponentially small for large $m$ and $t$
the asymptotic formula turns out to be very efficient. For an example
see Fig.~\ref{fig:compnum}. Second, the constant term $C(t,h)$, equation
(\ref{asyconstant}), does not depend on $\a$. For this reason it should
agree with the constant obtained by Barouch and McCoy \cite{BaMc71}
in form of infinite double products in their analysis of the static
correlation functions (see equations (3.17)-(3.19) of their paper).
We have numerical evidence that this is indeed the case.
\begin{figure}
\begin{center}
\includegraphics[width=.80\textwidth]{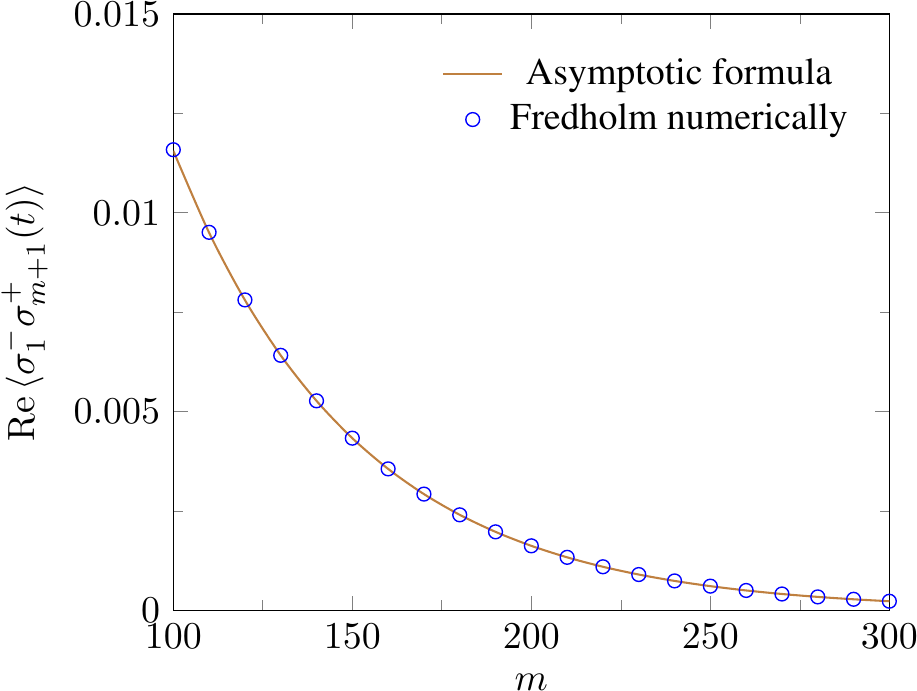}
\caption{\label{fig:compnum}
Real  part of  $\< \s^-_1 \sigma^+_{m+1}(t)\>$ as a function of $m$
for $T/J=0.05$, $h/J=0.1$ and $Jt=10$ evaluated numerically and from
(\ref{xxtransasy}).
}
\end{center}
\end{figure}

\vspace{.5ex}
\noindent {\bf Acknowledgements.}
The authors would like to thank Alexander Its and Nikita Slavnov for helpful
discussions. FG is supported by the Deutsche Forschungsgemeinschaft
within the framework of the research unit FOR 2316 `Correlations
in integrable quantum many-body systems'. The work of KKK is supported
by the CNRS and by the `Projet international de coop\'eration scientifique
No.\ PICS07877': \textit{Fonctions de corr\'elations dynamiques dans la
cha\^{\nodoti}ne XXZ \`a temp\'erature finie}, Allemagne, 2018-2020.
JS is supported by JSPS KAKENHI Grants, numbers 18K03452 and 18H01141.

\bibliographystyle{amsplain}
\bibliography{hub}

\end{document}